\documentclass[submission,copyright,creativecommons]{eptcs}
\providecommand{\event}{GASCom 2024} 

\usepackage{amsmath}
\usepackage{amssymb}
\usepackage{amsthm}
\usepackage[english]{babel}
\usepackage[T1]{fontenc}

\usepackage[linesnumbered, ruled, vlined]{algorithm2e} 

\usepackage{enumerate}

\usepackage{xcolor}
\definecolor{codegreen}{rgb}{0,0.6,0}
\definecolor{backcolour}{rgb}{0.95,0.95,0.92}

\usepackage{listings}
\lstdefinestyle{mystyle}{
	backgroundcolor=\color{backcolour},   
	commentstyle=\color{codegreen},
	keywordstyle=\color{magenta},
	numberstyle=\tiny\color{codegray},
	stringstyle=\color{codepurple},
	basicstyle=\ttfamily\footnotesize,
	breakatwhitespace=false,         
	breaklines=true,                 
	captionpos=b,                    
	keepspaces=true,                 
	numbers=left,                    
	numbersep=5pt,                  
	showspaces=false,                
	showstringspaces=false,
	showtabs=false,                  
	tabsize=2
}
\usepackage{url}

\newtheorem{cor}{Corollary}
\newtheorem{thm}{Theorem}
\theoremstyle{definition}
\newtheorem{dfn}{Definition}

\def\ub#1{
	\leavevmode\hbox{%
		\setbox0\hbox{$#1$}\dp0 0pt
		\vrule height.5ex width.4pt depth.33333ex \kern-.4pt
		\vtop{\hbox{\kern.15em \box0\kern.15em}\kern.33333ex \hrule}%
		\kern-.4pt \vrule height.5ex width.4pt depth.33333ex \kern-.4pt
	}%
}
\newcommand\bloc[2][]{\underset{#1}{\ub{#2}}\ }

\newcommand\N{\mathbb{N}}
\newcommand\OK{\mathbb{O}}

\usepackage{iftex}

\ifpdf
  \usepackage{underscore}         
  \usepackage[T1]{fontenc}        
\else
  \usepackage{breakurl}           
\fi

\title{Self-descriptive Sequences directed by two Periodic Sequences}
\author{
Shigeki Akiyama
\institute{Institute of Mathematics, University of Tsukuba \\1-1-1 Tennodai, Tsukuba, Ibaraki, 305-8571 Japan}
\email{akiyama@math.tsukuba.ac.jp}
\and
Damien Jamet 
\institute{Univ. Lorraine, Loria, UMR 7503\\Vand{\oe}uvre-l{\`e}s-Nancy, F-54506, France}
\email{damien.jamet@loria.fr}
\and
Ir\`ene Marcovici 
\institute{Univ Rouen Normandie, CNRS,\\ Normandie Univ LMRS UMR 6085,\\ F-76000 Rouen, France}
\email{irene.marcovici@univ-rouen.fr}
\and
Mai-Linh Trân Công 
\institute{\'Ecole Normale Supérieure de Lyon\\ 15 parvis René Descartes, F-69342 Lyon, France}
\email{mai-linh.tran\_cong@ens-lyon.fr}
}
\def\titlerunning{Self-descriptive Sequences directed by two Periodic Sequences}
\def\authorrunning{S.~Akiyama, D. Jamet, I. Marcovici, M.-L. Trân-Công}
\begin{document}
\maketitle

\section{Introduction}
A self-descriptive sequence $(u_n)_{n \in \mathbb{N}}$ is an infinite concatenation of finite powers of a letter (usually called runs) $(w_n)_{n \in \mathbb{N}}$ such that $|w_n| = u_n$ where $|x|$ denotes the length of the finite word $x$. The best known self-descriptive sequence is certainly the Oldenburger word $\OK_{1,2} = (k_n)_{n \in \N}$ \cite{RO1939,WK1966} defined by $u_0 = 1$, $w_{2n} = 1^{u_{2n}}$ and $w_{2n+1} = 2^{u_{2n+1}}$ for all $n \in \mathbb{N}$. Until recently, the Oldenburger word was still called the Kolakoski word in reference to \cite{WK1966}, but actually, it first appeared in \cite{RO1939}.

The Oldenburger word is a special case of a self-descriptive sequence. Indeed, the run $w_n$ is entirely determined by knowledge of its index $n$: its length is equal to $u_n$ and its single letter is determined by the parity of $n$.

In \cite{BJM23}, the authors focus on a larger family of self-descriptive sequences where the $w_n$'s are determined not only by their index $n$ but also by another sequence, namely the directing sequence of $u$. In practice, given a sequence $t = (t_n)_{n \in \mathbb{N}}$ on the alphabet $\mathcal{A} \in \{1,2,\dots\}$, the sequence directed by $t$ is the sequence $u$ defined by: $u = {t_0}^{u_0} {t_1}^{u_1} \cdots {t_n}^{u_n} \cdots$\, For example, the Oldenburger word $\OK_{1,2} = 1^{u_0}2^{u_1}1^{u_2}2^{u_4}\cdots$ is directed by the sequence $t = (12)^\omega$. 

One of the most fascinating questions about the sequence $\OK_{1,2}$ concerns the existence and possible value of the frequencies of occurrences of each of its letters \cite{Keane91}: \textit{Do the letters $1$ and $2$ have frequencies of occurrences $f_1$ and $f_2$ in $\OK_{1,2}$? If so, does $f_1=f_2=\frac{1}{2}$?}\, Recall that the frequency of occurrences of the letter $a$ in the sequence $u$ is the limit, when $n$ tends to $+\infty$, of the average number of $a$ in the prefix $u_0 \cdots u_{n-1}$ of $u$.

The notion of self-descriptive sequence is related to that of differentiable word and smooth word \cite{BBC05,BDLV06}. A sequence over $\mathcal{A}$ is differentiable over $\mathcal{A}$ if it is the infinite concatenation of runs whose lengths have values in $\mathcal{A}$. More precisely, a sequence $(u_n)_{n \in \mathbb{N}}$ over a finite alphabet $\mathcal{A} \subset \mathbb{N}$ is differentiable if there exist two sequences $(x_n)_{n \in \mathbb{N}}$ and $(\alpha_n)_{n \in \mathbb{N}}$ over $\mathcal{A}$, such that $u = x_0^{\alpha_0} x_1^{\alpha_1}x_2^{\alpha_2}\cdots$ with $x_n \neq x_{n+1}$ and $\alpha_n \neq 0$ for all $n \in \mathbb{N}$. The sequence $(\alpha_n)_{n \in \mathbb{N}}$ is the derivative sequence of $u$. Finally, the sequence $u$ is smooth if it is infinitely differentiable.

The sequence $\OK_{1,2}$ is a fixed point for differentiation. It is self-descriptive, differentiable and smooth over $\{1,2\}$. As with the sequence $\OK_{1,2}$, the question of the existence of frequencies of occurrences and their values in smooth words on the alphabet $\{1,2\}$ is still open. The first significant result on this question is due to V. Chv{\'a}tal \cite{chvatal93}: for $n$ large enough,
\begin{equation*}
0.49916 \leq \dfrac{|k_0 \cdots k_{n-1}|_1}{ n} \leq 0.50084,
\end{equation*} 
where $|w|_1$ denotes the number of occurrences of the letter $1$ in $w$.
These bounds have been slightly improved by M. Rao \cite{Rao12} using a method quite similar to Chv{\'a}tal's but with greater computing power:
\begin{equation*}
0.49992 \leq \dfrac{|k_0 \cdots k_{n-1}|_1}{ n} \leq 0.50008.
\end{equation*} 
This is the best currently known bound for $\OK_{1,2}$ and for the set of smooth words on the alphabet $\{1,2\}$. 
In other words, neither the existence nor the values of the frequencies of occurrences are known for any smooth word on the $\{1,2\}$ alphabet. On the other hand, over the $\{a,b\}$ alphabets where $a$ and $b$ have the same parity, it is possible to determine the frequencies of certain smooth words. This is the case at least for the Oldenburger word $\OK_{1,3}$ (resp. $\OK_{3,1}$) defined on the alphabet $\{1,3\}$ directed by $(13)^\omega$ (resp. by $(31)^\omega$) \cite{Sing04} and for the extreme smooth words (in the sense of lexicographic order) \cite{BJP08}.

Since the work of V. Chv{\'a}tal \cite{chvatal93} and M. Rao \cite{Rao12}, it is reasonable to expect that the frequencies of occurrences of each letter in $\OK_{1,2}$ are equal to $\frac{1}{2}$. In other words, it is reasonable to assume that the frequencies of occurrence in $\OK_{1,2}$ and in its directing sequence are identical.

As far as we know, none of the works on smooth words or on the Oldenburger word has shown the existence or the non-existence of frequencies of occurrences in a non-trivial deterministic self-descriptive sequence over the alphabet $\{1,2\}$. 

Similarly, none of these works has proved the existence of a non-trivial deterministic self-descriptive sequence that shares (resp. does not share) its frequencies of occurrences with its directing sequence.

In the present work, we exhibit a class of self-descriptive sequences that can be explicitly computed and whose frequencies are known. In particular, as a corollary of our main result, we prove that the sequence introduced in \cite{BJM23} has the expected frequencies of occurrences.

\section{Definitions and basic notions}
Let $\mathcal{A}$ be a finite alphabet. The set of finite words over $\mathcal{A}$ is denoted by $\mathcal{A}^\star$. If $w = w_0 \cdots w_k \in \mathcal{A}^\star$ is a finite word over the alphabet $\mathcal{A}$ with $w_i \in \mathcal{A}$ for $i=0,1,\cdots,k$. Let $|w|$ stand for the \textbf{length} of $w$, that is the number of letters occurring in $w$. If $w = w_0 \cdots w_k$, then $|w|=k+1$. In particular, $|\varepsilon|=0$. Let $w \in \mathcal{A}^\star$ and let $a \in \mathcal{A}$. We set $|w|_a = \# \left\{i \in\{ 0,1,\dots, |w|-1\} \, |\ w_i = a\right\}$.

\begin{dfn}[Self-descriptive sequence\label{def::def1}]
	Let $\mathcal{A} \subset \mathbb{N}^\star$ be a finite alphabet. The infinite sequence $u = (u_n)_{n \in \N} \in \mathcal{A}^\mathbb{N}$ is said to be \textbf{self-descriptive} if there exists a sequence $\delta=(\delta_n)_{n \in \mathbb{N}}$ over $\mathcal{A}$ such that
	\begin{equation}\label{eq::auto}
	u = \delta_0^{u_0} \delta_1^{u_1}  \delta_2^{u_2}  \cdots \delta_n^{u_n}  \cdots  
	\end{equation}
	The sequence $\delta$ is called the \textbf{directing sequence} of $u$ and one says that $u$ is directed by $x$.	
\end{dfn}

In other words, the sequence $u$ is self-descriptive if it is the concatenation of runs of size $u_0$, $u_1$, $u_2$, \ldots\, respectively. Note that in the definition of self-descriptive sequences, unlike that of differentiable or smooth words, it is not necessary that $x_n \neq x_{n+1}$. Furthermore, if $0 \notin \mathcal{A}$, then the sequence $(x_n)_{n \in \mathbb{N}}$ is entirely determined by $u$. In other words, there exists a canonical bijection between sequences over $\mathcal{A}$ and self-descriptive sequences over $\mathcal{A}$.

Let $u = \delta_0^{u_0} \delta_1^{u_1}  \delta_2^{u_2}  \cdots \delta_n^{u_n} \cdots $ be a self-descriptive sequence over the alphabet $\{1,2\}$. Let $(m_k)_{k \in \mathbb{N}}$ (resp. $(n_k)_{k \in \mathbb{N}}$) be the increasing sequence over $\mathbb{N}$ such that $u_{i} = 1$ (resp. $u_i=2$) if and only if there exists $k \in \mathbb{N}$ such that $i=m_k$ (resp. $i=n_k$). In other words, $(m_k)_{k \in \mathbb{N}}$ (resp. $(n_k)_{k \in \mathbb{N}}$) is exactly the ordered sequences of the indices where $u$ is equal to $1$ (resp. $2$). Let $T_1 =\left(\delta_{m_k}\right)_{k \in \{1,2\}^\mathbb{N}}$ and $T_2 =\left(\delta_{n_k}\right)_{k \in \{1,2\}^\mathbb{N}}$. 

In the present work, since we are only interested in frequencies of letters, we assume, without loss of generality, that $u_0=u_1=2$. The sequence $u$ and its directing sequence $\delta$ are then computable from $T_1$ and $T_2$ as follows: 
\begin{center}
	\begin{minipage}{0.625\linewidth}	
		\begin{lstlisting}[caption={\texttt{Python} function computing $u$ and $\delta$ from $T_1$ and $T_2$.},language=Python, label={lst::liste1}]
def OK(T1, T2):
	u = [2,2]
	delta = [2]
	k = 1
	while len(T1) > 0 and len(T2) > 0:
		if u[k] == 1:
			c = T1.pop(0)
			u += [c]
		else:
			c = T2.pop(0)
			u += [c] * u[k]
		delta += [l]
		k += 1
	return u,delta
\end{lstlisting}
	\end{minipage}
\end{center}
One then says that $u$ is also directed by sequences $T_1$ and $T_2$. For instance, if $T_1 = \textcolor{orange}{121\cdots}$ and $T_2 = \textcolor{codegreen}{12\cdots}$, then 
\begin{eqnarray}
u & = & \bloc[2]{22}  
\cdot \bloc[2]{\textcolor{codegreen}{11}}
\cdot \bloc[1]{\textcolor{orange}{1}}
\bloc[1]{\textcolor{orange}{2}}
\cdot \bloc[1]{\textcolor{orange}{1}}
\bloc[2]{\textcolor{codegreen}{22}} \cdots \\
&	=	&	2^2 \cdot \textcolor{codegreen}{1}^2 \cdot \textcolor{orange}{1}^1 \textcolor{orange}{2}^1 \cdot \textcolor{orange}{1}^1  \textcolor{codegreen}{2}^2 \cdots \\
& = & \bloc[2]{22}\label{eq::eq1}  
\cdot \bloc[w_0]{\textcolor{codegreen}{11}}
\cdot \bloc[w_1]{\textcolor{orange}{1}\textcolor{orange}{2}}
\cdot \bloc[w_2]{\textcolor{orange}{1}\textcolor{codegreen}{22}} \cdots =	22 \hspace{0.18cm} \cdot \hspace{0.05cm} w_0 \hspace{0.105cm} \cdot \hspace{0.1cm} w_1 \hspace{0.2125cm} \cdot \hspace{0.15cm} w_2 \hspace{0.3cm} \cdots 		
\end{eqnarray}
with $w_i \in \{a,b,c,d\}^*$, for $i \in \N$, $a= \textcolor{orange}{1}$, $b= \textcolor{orange}{2}$, $c= \textcolor{codegreen}{1}$ and $d= \textcolor{codegreen}{1}$.

\section{Main result}
The main result of the present work is:
\begin{thm}[Main result\label{thm::thmMain}]
	Let $u \in \{1,2\}^\mathbb{N}$ be a sequence over $\{1,2\}$ directed by two periodic sequences $T_1 = (x_1)^\omega \in \{1,2\}^\mathbb{N}$ and $T_2 = (x_2)^\omega \in \{1,2\}^\mathbb{N}$, with $x_1,x_2 \in \{1,2\}^*$. Let $p_1 = \dfrac{|x_1|_1}{|x_1|}$ and $q_2 = \dfrac{|x_2|_2}{|x_2|}$. One has
	$$f_1:= \lim_{n \to \infty} \dfrac{|u_0 \cdots u_{n-1}|_1}{n} = \frac{(1 - q_2)(p_1 + 2q_2 +\sqrt{\Delta})}{2 + \sqrt{\Delta} - p_1}$$
	with $\Delta = (p_1 + 2q_2)^2 - 8(p_1 + q_2 - 1)$.
	
	If  $\delta = (\delta_n)_{n \in \mathbb{N}} \in \{1,2\}$ is directing $u$, then
	$$\lim_{n \to \infty} \dfrac{|\delta_0 \cdots \delta_{n-1}|_1}{n} = p_1 f_1 + p_2 (1-f_1).$$
\end{thm}

\begin{proof}[Sketch of proof]
	Let us recode $T_1$, $T_2$ and $u$ over $\{a,b,c,d\}$ as follows: rewrite $T_1$ (resp. $T_2$) as the image of $T_1$ (resp. $T_2$) by the morphism $1 \mapsto a$, $2 \mapsto b$ (resp. $1 \mapsto c$, $2 \mapsto d$).
	In $u$, let us substitute $a$ (resp. $b$) for isolated $1$'s (resp. isolated $2$'s) and $cc$ (resp. $dd$ ) for double $1$'s (resp. double $2$'s).
	
	\begin{enumerate}
		\item Let $u = (u_n)_{n \in \N} = 22 \cdot w_0 \cdot w_1 \cdot \cdots \cdot w_n \cdot \cdots$, where $w_{n+1}$ is the image of $w_n$ by the recoding rule along $T_1$ and $T_2$ (see (\ref{eq::eq1}) for an example).
		\item Let  Let $p_2=1-p_1$ and $q_1=1-q_2$.	Let 
		$$	A = 
		\begin{pmatrix}
		p_1 & 0 & p_1 & 0 \\
		p_2 & 0 & p_2 & 0 \\
		0 & 2q_1 & 0 & 2q_1 \\
		0 & 2q_2 & 0 & 2q_2
		\end{pmatrix} \text{ and } v_n = 
		\begin{pmatrix}
		|w_n|_a \\ |w_n|_b \\ |w_n|_c \\ |w_n|_d 
		\end{pmatrix} \qquad \implies v_{n+1} = A \cdot v_n + e_n,$$ 
		where $e_n$ is an "error" vector and is \underline{bounded}.
		\item $A$ is primitive with exactly two eigenvalues $0 < |\alpha_2| \leq 1 < \alpha_1$. By the Perron-Frobenius theorem, there exists a right (resp. left) eigenvector vectors $\mathbf{r}$ (resp. $\mathbf{\ell}$) of $A$ such that:
		$\mathbf{\ell} \mathbf{r}  = 1$ and $\lim_{n \to \infty}{\alpha^{-n}_1 A^n} =  \mathbf{r} \cdot {~}^t \mathbf{\ell}$.
		\item One cuts the sequence $u$ into words $(g_n)_{n \in \mathbb{N}}$ following Algorithm \ref{algo::algo1} and shows:
		\begin{itemize}
			\item$|w_0 \cdots w_{\ell_n}| = o(|g_n|)$ and $\displaystyle \lim_{n \to \infty}|g_n| = + \infty$
			\item $\dfrac{|u_0 \cdots u_n|_a}{|u_0 \cdots u_n|}  =  \dfrac{|w_0 \cdots w_{l_n -1}|_a + |g_n|_a}{|w_0 \cdots w_{l_n -1}| + |g_n|} =	\dfrac{ \frac{|g_n|_a}{|g_n|} + o(1)}{1 + o(1)} \underset{n \to \infty}{\longrightarrow} r_0$
		\end{itemize}
		\begin{algorithm}[h]
			\DontPrintSemicolon
			\KwIn{$u = w_0 w_1 \cdots w_n \cdots = u_0 u_1 \cdots $}
			$\ell_0 \gets 0$ \tcp*{initial left index}		
			\For{each $n \in \mathbb{N}$}{
				$g_n \gets u_{\ell_n} \cdots u_n$ \tcp*{$|g_n| = (n+1)-|w_0 \cdots w_{\ell_n-1}|$}
				\uIf{$ |g_n| +1 > |w_0 \cdots w_{\ell_n}|^2$ }{
					$\ell_{n+1} \gets \ell_n + 1$ \tcp*{increment left index}
				}
				\Else{
					$\ell_{n+1} \gets \ell_n$ \tcp*{keep left index}
				}
			}
			\caption{Cutting the sequence $\mathcal{O}$ into $(g_n)_{n \in \mathbb{N}}$\label{algo::algo1}}
		\end{algorithm}
	\end{enumerate}
\end{proof}

As a direct consequence of Theorem \ref{thm::thmMain}, we prove the existence of the frequencies in the sequence introduced in Section 5 of \cite{BJM23}:
\begin{dfn}[BJM sequence \cite{BJM23}]
	Let $U = (u_n)_{n \in \N}$ be the self-descriptive sequence $U = x_0^{u_0} x_1^{u_1}x_2^{u_2}\cdots$ defined by $x_0=u_0=u_1=2$, and for all $n \in \mathbb{N}^\star$:
	\begin{enumerate}[i)]
		\item if $u_n=1$, then $x_n=1$ (resp. $x_n=2$) if $|u_0 \cdots u_n|_1$ is odd (resp. even),
		\item if $u_n=2$ then $x_n=1$.
	\end{enumerate}
	In other words, the runs of size 2 (except the first one) are filled by 1, and the runs of size 1 are filled alternatively by 1 and 2. The sequence $X = (x_n)_{n \in \N}$ is the directed sequence of $U$.
\end{dfn}
In \cite{BJM23}, the authors showed that the frequencies of occurrence in $U$ cannot be equal to those of its directing sequence $X$. However, the authors did not prove the existence of the frequencies but only that they cannot be identical... if they exist. Since $U$ is directed by $T_1 = (12)^\omega$ and $T_2 = 1^\omega$, it directly follows from Theorem \ref{thm::thmMain} that:
\begin{cor}
	Let $U$ be the sequence directed by $T_1 = (12)^\omega$ and $T_2 = 1^\omega$. Then
	$$\lim_{n \to \infty} \dfrac{|U_0 \cdots U_{n-1}|_1}{n} = \dfrac{7-\sqrt{17}}{4} \text{ and }\lim_{n \to \infty} \dfrac{|X_0 \cdots X_{n-1}|_1}{n} = \dfrac{1+\sqrt{17}}{8}.$$
\end{cor}
\begin{proof}
	In that present case, $p_1 = 0.5$ and $q_2=0$.
\end{proof}

\section{Conclusion and perspectives}
In the present work, we have shown that self-descriptive sequences directed by two periodic sequences have frequencies. We have also given an explicit expression for these frequencies.

In future work, it will be interesting to extend this result to non-periodic sequences. For example, Sturmian words, namely the aperiodic sequences with the least number of finite factors, are good candidates for directing sequences.


\bibliographystyle{eptcs}
\bibliography{mybiblio}
\end{document}